\documentclass[runningheads,a4paper]{llncs}
\usepackage[english]{babel}
\usepackage[T1]{fontenc}
\usepackage[utf8]{inputenc}
\usepackage{algorithm}
\usepackage[noend]{algpseudocode}
\usepackage{amsmath,amssymb}
\usepackage{caption}
\usepackage{hyperref}
\usepackage{makeidx}
\usepackage{multirow}
\usepackage{ragged2e}
\usepackage{rotating}
\usepackage{siunitx}
\usepackage[labelformat=simple]{subcaption}
\usepackage{tabularx}
\usepackage{todonotes}
\usepackage{tikz}
\usepackage{sidecap,wrapfig}

\newcolumntype{C}{>{\setlength\hsize{1\hsize}\centering}X}
\newcolumntype{R}{>{\raggedleft\arraybackslash}X}

\usetikzlibrary{arrows,calc,fit,shapes}

\title{A Sidetrack-Based Algorithm For Finding the $k$ Shortest Simple Paths in a Directed Graph}
\author{Denis Kurz, Petra Mutzel}
\institute{Department of Computer Science, TU Dortmund, Germany\\
\email{\{denis.kurz,petra.mutzel\}@tu-dortmund.de}}

\newcommand{\keywords}[1]{\par\addvspace\baselineskip
\noindent\keywordname\enspace\ignorespaces#1}

\newcommand{\dev}{\text{dev}}
\newcommand{\dist}{d}
\newcommand{\pref}{\text{pref}}
\newcommand{\suff}{\text{suff}}

\newcommand{\norm}[1]{\left|#1\right|}

\begin{document}
\def\sectionautorefname{Section}
\def\figureautorefname{Figure}
\def\tableautorefname{Table}
\def\theoremautorefname{Theorem}

\mainmatter
\maketitle

\begin{abstract}
We present an algorithm for the $k$ shortest simple path problem on weighted directed graphs ($k$SSP) that is based on Eppstein's algorithm for a similar problem in which paths are allowed to contain cycles.
In contrast to most other algorithms for $k$SSP, ours is not based on Yen's algorithm \cite{CUSTOM:journals/networks/Yen71} and does not solve replacement path problems.
Its worst-case running time is on par with state-of-the-art algorithms for $k$SSP.
Using our algorithm, one may find $\mathcal O(m)$ simple paths with a single shortest path tree computation and $\mathcal O(n + m)$ additional time per path in well-behaved cases, where $n$ is the number of nodes and $m$ is the number of edges.
Our computational results show that on random graphs and large road networks, these well-behaved cases are quite common and our algorithm is faster than existing algorithms by an order of magnitude.
Further, the running time is far better predictable due to very small dispersion.
\end{abstract}

\section{Introduction}
\label{sec:introduction}

The \emph{$k$ shortest path problem} in weighted, directed graphs ($k$SP) asks for a set of $k$ paths from a \emph{source} $s$ to a target \emph{target} $t$ in a graph with $n$ nodes and $m$ edges.
Every path that is not output by an algorithm should be at least as long as any path in the output.
Algorithms for this problem can be useful tools when it is hard to specify constraints that a solution should satisfy.
Instead of computing only one shortest path, $k$SP algorithms generate $k$ paths, and the user can then pick the one that suits their needs best.
The best known algorithm for this problem runs in time $\mathcal O(m + n \log n + k \log k)$ and is due to Eppstein~\cite{DBLP:journals/siamcomp/Eppstein98}.
In the initialization phase, the algorithm builds a data structure that contains information about all $s$-$t$ paths and how they interrelate with each other, in time $\mathcal O(m + n \log n)$.
This can even be reduced to $\mathcal O(m + n)$ if the shortest path tree (SP tree) is given in the input or if the SP tree can be computed in time $\mathcal O(m + n)$.
In the enumeration phase, a \emph{path graph} is constructed.
The path graph is a quaternary min-heap where every path starting in the root correlates to an $s$-$t$ path in the original graph.
We require $\mathcal O(k \log k)$ time for the enumeration phase if we want the output paths to be ordered by length.
If the order in which the paths are output does not matter, Frederickson's heap selection algorithm~\cite{DBLP:journals/iandc/Frederickson93} can be used to enumerate the paths after the initialization phase in time $\mathcal O(k)$.

The \emph{$k$ shortest simple path problem} ($k$SSP), introduced by Clarke, Krikorian and Schwartz~\cite{CUSTOM:journals/jsiam/ClarkeKR63}, seems to be more expensive, computationally.
In contrast to $k$SP, the computed paths are required to be simple, i.e., they must not contain a cycle.
The extra effort may be well-invested if many of the $k$ shortest paths are non-simple and we are only interested in simple paths.
The algorithm by Yen~\cite{CUSTOM:journals/networks/Yen71} used to have the best theoretical worst-case running time of $\mathcal O(kn(m + n \log n))$ for quite some time.
Gotthilf and Lewenstein~\cite{DBLP:journals/ipl/GotthilfL09} improved upon this bound recently.
They observed that $k$SSP can be solved by solving $\mathcal O(k)$ all pairs shortest path (APSP) instances.
Using the APSP algorithm by Pettie~\cite{DBLP:journals/tcs/Pettie04}, they obtain a new upper bound of $\mathcal O(kn(m + n \log \log n))$.
Vassilevska Williams and Williams~\cite{DBLP:conf/focs/WilliamsW10} showed that, for constant $k$, an algorithm for $k$SSP with running time $\mathcal O(n^{3-\varepsilon})$ for some positive $\varepsilon$ (\emph{truly subcubic}) would also yield algorithms with truly subcubic running times for some other problems, including APSP.
A recent survey of the field is due to Eppstein~\cite{CUSTOM:arxiv/Eppstein14}.
The $k$SSP on undirected graphs seems to be significantly easier.
Katoh et al.~\cite{DBLP:journals/networks/KatohIM82} proposed an algorithm that solves $k$SSP on undirected graphs in time $\mathcal O(k(m + n \log n))$.

A subproblem occurring in Yen's algorithm is the (restricted) \emph{replacement path problem}.
Given a shortest $s$-$t$ path $p$ in a graph, it asks for a set of paths as follows.
For each $i < \norm{p}$, the set has to include a shortest simple path that uses the first $i-1$ edges of $p$, but not the $i$th.
This problem has to be solved $\mathcal O(k)$ times to find the $k$ shortest simple paths using Yen's algorithm.
In the original version of Yen's algorithm, the replacement paths are found using $\mathcal O(\norm{p})$ shortest path computations, resulting in time $\mathcal O(n(m + n \log n))$.
Hershberger et al.~\cite{DBLP:journals/talg/HershbergerMS07} compute one SP tree rooted in $s$ and one reversed SP tree rooted in $t$, respectively.
They use these two trees to find a replacement path in constant time per edge on $p$, cutting down the time required to find all replacement paths to $\mathcal O(m + n \log n)$ when Dijkstra's algorithm is used.
However, the paths generated this way are not guaranteed to be simple.
Such non-simple paths can be detected in constant time and repaired by falling back to Yen's replacement path computation for the path edge in question.
Since they do not provide an upper bound for the number of non-simple paths that may occur using this method, the worst-case running time is again $\mathcal O(n(m + n \log n))$.

Some approaches reuse one fixed reversed SP tree $T_0$ rooted in $t$ and computed during the initialization of their $k$SSP algorithm, in contrast to $\mathcal O(1)$ SP trees per replacement path instance.
Pascoal~\cite{Pascoal06implementationsand} noticed that the replacement path that deviates from $p$ at node $v$ might be one that uses an edge $(v, w)$ to an unused successor $w$ of $v$ and then follows the path from $w$ to $t$ in $T_0$.
Therefore, they test whether the shortest such path is simple, and fall back to a full shortest path computation if it is not.
Although they do not describe in detail how this check is done, it can be done in time $\mathcal O(m + n)$ per replacement path instance by partitioning the nodes into blocks as described by Hershberger et al.~\cite{DBLP:journals/talg/HershbergerMS07}.
Feng~\cite{DBLP:journals/networks/Feng14} uses the reversed SP tree to partition $V$ into three classes.
For each edge $(u, v)$ on $p$ for which we want to compute a replacement path, \emph{red} nodes have already been used to reach $v$ via $p$.
A \emph{yellow} node $v$ is a non-red upstream node of some red node in $T_0$, i.e., the path from $v$ to $t$ in $T_0$ contains a red node.
All other nodes are \emph{green}.
They then do shortest path computations from $v$ using Dijkstra's algorithm like Yen.
However, they are able to restrict the search to yellow nodes, resulting in a significantly smaller search space.
Feng does not provide upper bounds on the size of this search space, resulting again in a worst-case running time of $\mathcal O(n(m + n \log n))$ for each replacement path instance.

\emph{Our contribution.}
We propose an algorithm that was derived from Eppstein's notion of a path graph~\cite{DBLP:journals/siamcomp/Eppstein98}.
Our algorithm achieves the same worst-case running time as Yen's algorithm.
Like Yen, we rely on shortest path (tree) computations.
In contrast to Yen-based algorithms, however, our algorithm may draw $\mathcal O(m)$ simple paths from one shortest path tree computation.
If the underlying graph is acyclic, the revised algorithm at the end of this paper requires $\mathcal O(n \log n + k(m + n))$ without further modifications.
Alternatively, one could test whether the graph is acyclic and then use Eppstein's algorithm.
However, this method fails if the graph has just a single cycle, in which case our algorithm appears to be a good choice.
Our algorithm works on multigraphs without modification.
This is also true for every other $k$SSP algorithm we know of.
After some definitions in~\autoref{sec:definitions}, we propose a simplified version of our algorithm with running time $\mathcal O(km(m + n \log n))$ in~\autoref{sec:simple-algo}.
In~\autoref{sec:improve}, we show how this running time can be reduced to $\mathcal O(kn(m + n \log n))$, and how to reduce the number of shortest path tree computations in practice.
Finally, we present the results of our computational studies in~\autoref{sec:experiments} to prove the efficiency of our algorithm.

\section{Definitions}
\label{sec:definitions}

Let $G = (V, E)$ be a directed graph with \emph{node} set $V$ and \emph{edge} set $E$.
Let $s, t \in V$ be \emph{source} and \emph{target} nodes, respectively.
We assume an implicit edge weight function $c: E \to \mathbb R^+_0$ throughout this paper.
We denote the number of nodes $\norm V$ by $n$ and the number of edges $\norm E$ by $m$.
A \emph{path} connecting $v$ to $w$ in $G$, or $v$-$w$ path, is an edge sequence $p = (e_1, e_2, \ldots, e_r)$, $e_i = (v_i, w_i)$, with $v = v_1$, $w = w_r$ and $w_i = v_{i + 1}$ for $1 \le i < r$.
For the sake of simplicity, we only consider combinations of $G$, $s$ and $t$ such that there exists an $s$-$v$ path and a $v$-$t$ path in $G$ for every $v \in V$.
A node $u$ is said to be on the path $p$, denoted by $u \in p$, if $u = w$ or $u = v_i$ for some $i$.
If $v_i \neq v_j \neq w$ for $1 \leq i, j \leq r$, $p$ is a \emph{simple} path.
The \emph{prefix} $(e_1, \ldots, e_i)$ is a $v$-$w_i$ path and denoted by $p^i$.
The \emph{length} $c(p)$ of the path $p$ is the sum of edge weights of its edges.
If every $v$-$w$ path is at least as long as $p$, it is called a \emph{shortest $v$-$w$ path}.
We write $G - p$ to denote the induced subgraph $G[\{ v \in V \mid v \notin p\}]$.

The \emph{$k$ shortest simple path problem} ($k$SSP) is an enumeration problem.
Given a directed graph $G = (V, E)$ with source node $s \in V$, target node $t \in V$, edge weights $c$, and some $k \in \mathbb N$, we want to compute a set $P$ comprising $k$ simple paths from $s$ to $t$ in $G$ such that $c(p) \le c(p')$ for every pair $p \in P$, $p' \notin P$ of simple paths.
We obtain the \emph{$k$ shortest path problem} ($k$SP) if we do not require the computed paths to be simple.

A \emph{shortest path tree} (SP tree) $T$ of $G$ is a subtree of $G$ with node set $V$ such that each $v \in V$ has exactly one outgoing edge, which lies on a shortest $v$-$t$ path, or no outgoing edges if no $v$-$t$ path in $G$ exists.
We denote the latter case by $v \notin T$ even if $v$ is in the node set of $T$.
Our algorithm will compute several SP trees, the first of which we call \emph{initial SP tree} $T_0$.
An edge $e \notin T$ is called \emph{sidetrack} w.r.t. $T$; we will omit $T$ in most cases.
For a sidetrack $e = (v, w)$, the \emph{sidetrack cost} $\delta_T(e)$ is defined as $(c(e) + \dist(w)) - \dist(v)$, where $\dist(u)$ is the length of the unique $u$-$t$ path in $T$.
The sidetrack cost is therefore the difference between the length of a shortest $v$-$t$ path and the length of a shortest $v$-$t$ path that starts with $e$.
The sidetrack set $D_T(v)$ of a node $v \in V$ is the set of all sidetracks w.r.t. $T$ with tails on the unique $v$-$t$ path in $T$.
When sidetracks are organized in heaps, we use sidetrack costs for comparison.

Let $p = (e_1, \ldots, e_k)$, $p' = (f_1, \ldots, f_l)$ be two $s$-$t$ paths, and $i^* = \max \{ i \mid e_j = f_j \text{ for } 1 \le j < i \}$.
Then, with respect to $p$, $i^*$ is the \emph{deviation index}, the tail of $e_{i^*}$ is the \emph{deviation node} $\dev(p')$, and $e_{i^*}$ is the \emph{deviation edge} of $p'$.
As is quite usual for $k$SSP algorithms, we will discover paths in a hierarchical fashion: a path $p'$ is added to the candidate set after $p$ was extracted from the candidate set.
In such cases, $p$ is called the \emph{parent path} of $p'$.
When $p$ is omitted, the terms deviation node and edge are w.r.t. the parent path of $p'$.
By removing the deviation edge of $p$ from $p$, $p$ is split into its \emph{prefix path} $\pref(p) := p^{i^*}$ starting in $s$, and its \emph{suffix path} $\suff(p)$ ending in $t$.
The initial $s$-$t$ path $p_0$ in $T_0$ has no parent path and thus no deviation edge.
We define its suffix path to be $p_0$ itself.

We introduce a generalization of Eppstein's representation \cite{DBLP:journals/siamcomp/Eppstein98} for paths.
Eppstein represented paths as sequences of sidetracks, which were all sidetracks w.r.t. the same shortest path tree.
In our representation, every sidetrack $e$ in a sidetrack sequence may be associated with a different shortest path tree $T_e$.
The path represented by a sidetrack sequence $(e_1, \ldots, e_r)$ can then be reconstructed as follows.
Starting in $s$, we follow the initial SP tree $T_0$ until we reach the tail of $e_1$.
After reaching the tail of $e_i$, we traverse $e_i$ and follow $T_{e_i}$ until we reach the tail of $e_{i + 1}$, or, in case $i = r$, until we reach $t$.
Note that Eppstein's representation is the special case where $T_e = T_0$ for each $e$ in a sidetrack sequence, and that both Eppstein's sidetrack sequences and our generalized ones may represent non-simple paths.
The distance from a node $v$ to $t$ in a shortest path tree $T_e$ associated with a sidetrack $e$ is denoted by $\dist_e(v)$.

\section{Basic Algorithm}
\label{sec:simple-algo}

In this section, we propose a rather simple way to enumerate the $k$ shortest simple paths.
We will describe later how to modify this algorithm to achieve our proclaimed running time guarantee.

We initialize an empty priority queue $Q$ that is going to manage candidate paths.
The key of a path in $Q$ is its length.
We compute the initial shortest path tree $T_0$ and push its unique $s$-$t$ path, represented by an empty sidetrack sequence, to $Q$.
We now process the paths in $Q$ in order of increasing length until we found $k$ simple paths.

Let $(e_1, \ldots, e_r)$ be a sidetrack sequence extracted from $Q$, and $p$ the path that is represented by this sequence.
Although the first path that is pushed to $Q$ is always simple, we will eventually push non-simple paths to $Q$, too.
Therefore, we first have to determine whether $p$ is simple in a \emph{pivot step}.
This check can be done by simply walking $p$ and marking every visited node.

We first describe how to handle the simple case.
We start by outputting $p$.
For every sidetrack $e = (u, v)$ with $u \in \suff(p)$, we discover a new path $p'$ represented by the sequence $(e_1, \ldots, e_r, e)$.
We set $\dev(p') = u$, and push $p'$ to $Q$.
The length of $p'$ can easily be computed as $c(p) + \delta_{T_e}(e)$.
If $\dist_{e_r}(v)$ is undefined because $T_{e_r}$ does not contain a $v$-$t$ path, we simply ignore $e$.
By choosing $T_e = T_{e_r}$, we simply reuse the shortest path tree that is also associated with the last sidetrack in the sequence representing $p$.
Note that sidetracks emanating from $t$ can safely be ignored.

\begin{figure}[tb]
\centering
\begin{subfigure}[b]{0.59\textwidth}
\centering
\begin{tikzpicture}[auto=right,every node/.style={minimum size=0.65cm}]
\draw (180:3cm) node[draw,circle] (s) {$s$};
\draw (135:1.73cm) node[draw,circle] (v1) {$v_1$};
\draw (225:1.73cm) node[draw,circle] (v2) {$v_2$};
\draw (45:1.73cm) node[draw,circle] (v3) {$v_3$};
\draw (-45:1.73cm) node[draw,circle] (v4) {$v_4$};
\draw (0:3cm)   node[draw,circle] (t) {$t$};

\draw[->,ultra thick] (s) -- (v1);
\draw[->,dashed] (s) -- node {$a$} (v2);
\draw[->,ultra thick] (v1) -- (v3);
\draw[->,ultra thick] (v2) -- (v1);
\draw[->,dashed] (v2) -- node {$b$} (v4);
\draw[->,dashed] (v3) -- node {$c$} (v2);
\draw[->,ultra thick] (v3) -- (t);
\draw[->,ultra thick] (v4) -- (v3);
\draw[->,dashed] (v4) -- node{$d$} (t);
\end{tikzpicture}
\caption{Example graph}
\label{fig:basic-example-graph}
\end{subfigure}
\begin{subfigure}[b]{0.39\textwidth}
\centering
\begin{tikzpicture}
\draw (0.1875, 0) node (empty) {$()$};
\draw (-1.125,-1.2) node (a) {$(a)$};
\draw (1.5,-1.2) node (c) {$(c)$};
\draw (1.5,-2.4) node (cprime) {$(c')$};
\draw (-2,-2.4) node (ab) {$(a, b)$};
\draw (-0.25,-2.4) node (ac) {$(a, c)$};
\draw (-2.7,-3.6) node (abc) {$(a, b, c)$};
\draw (-1.3,-3.6) node (abd) {$(a, b, d)$};

\draw[->] (empty) -- (a);
\draw[->] (empty) -- (c);
\draw[->] (a) -- (ab);
\draw[->] (a) -- (ac);
\draw[->] (ab) -- (abc);
\draw[->] (ab) -- (abd);
\draw[->] (c) -- (cprime);
\end{tikzpicture}
\caption{Sidetrack sequences}
\label{fig:basic-example-heap}
\end{subfigure}
\caption{Example for the basic algorithm. In \autoref{fig:basic-example-graph}, the thick, solid edges belong to $T_0$. In \autoref{fig:basic-example-heap}, every sidetrack is associated with $T_0$ except for $c'$, which is associated with the SP tree $T_1$ comprising the edges $b$ and $d$. An arrow from sequence $p$ to sequence $p'$ indicates that $p$ is the parent path of $p'$.}
\label{fig:basic-example}
\end{figure}
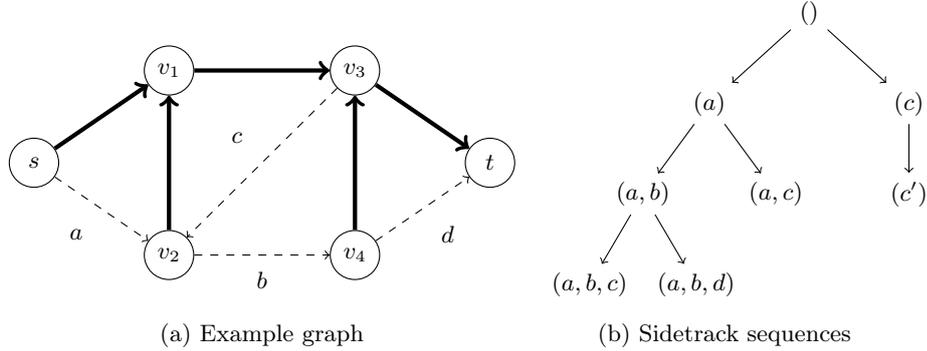

Consider the example in \autoref{fig:basic-example}.
The sidetrack sequence $(a)$ with $T_a = T_0$ represents a simple path $p$ that passes the nodes $s, v_2, v_1, v_3, t$ in this order.
The suffix of this path is its $v_2$-$t$ sub-path, and the sidetracks $b$ and $c$ have tails on this suffix.
Therefore, when $(a)$ is extracted from $Q$, $p$ is output and the sequences $(a, b)$ and $(a, c)$ with $T_a = T_b = T_c$ are pushed to $Q$.

Now assume we extracted a non-simple path $p$ represented by the sidetrack sequence $(e_1, \ldots, e_r)$.
We try to extend the concatenation of the prefix path of $p$ and $e_r$ to a simple $s$-$t$ path.
Let $e_r = (v, w)$.
Any valid extension has to avoid the nodes of $\pref(p)$ after $v$, and we are only interested in shortest extensions.
Therefore, we compute a new SP tree $T$ and distances $\dist$, but in $G - \pref(p)$ instead of $G$ to make sure that nodes of the prefix path of $p$ are not used again.
If $w \notin T$, $\pref(p)$ cannot be extended to a simple $s$-$t$ path, and we simply discard $p$.
Otherwise, we push the sequence $(e_1, \ldots, e_r)$ to $Q$ again.
In this new sequence, however, we associate $T$ with $e_r$ instead of $T_{e_r}$ from the old sequence.
The sequence represents a path $p'$ obtained by concatenating the simple prefix path of $p$, the edge $e_r$, and the $w$-$t$ path in $T$ that, by construction, avoids all nodes of $\pref(p)$.
The suffix itself is simple, too, because it is a shortest path in a subgraph of $G$.
Hence, $p'$ is simple.
The length of this path is $c(\pref(p)) + c(e_r) + \dist(w)$.

Consider again the example in \autoref{fig:basic-example}.
The sidetrack sequence $(a, c)$ with $T_a = T_c = T_0$ represents a non-simple path $p$ that visits the nodes $s$, $v_2$, $v_1$, $v_3$, $v_2$, $v_1$, $v_3$, $t$ in this order.
The deviation node of $p$ is $v_3$, its deviation edge $c$, and its prefix path is $(a, (v_2, v_1), (v_1, v_3))$.
We compute a new SP tree $T$ in $G - \pref(p)$, which only consists of the edge $d$.
Therefore, $T$ does not contain a $v_2$-$t$ path, and $p$ is discarded.

In contrast, assume the sequence $(c)$ with $T_c = T_0$ was just extracted from $Q$.
It represents almost the same path as the sequence above, but it skips the first visit of $v_2$.
Again, $v_3$ is the deviation node and $c$ the deviation edge.
The prefix path comprises the nodes $s$, $v_1$ and $v_3$.
After removing them temporarily, a new shortest path tree $T_1$ is computed, consisting only of the edges $b$ and $d$.
The sequence $(c)$ with $T_c = T_1$ is pushed to $Q$.
This new sequence represents the simple path $((s, v_1), (v_1, v_3), c, b, d)$, i.e., the concatenation of the prefix path of $p$, the last sidetrack $c$ in the extracted sequence, and the unique $v_2$-$t$ path in $T_1$.

Finally, when $(c)$ with $T_c = T_1$ is extracted, the represented path is output.
The sidetracks emanating from its prefix are $(v_2, v_1)$ and $(v_4, v_3)$.
Since $v_1, v_3 \notin T_1$, these sidetracks are ignored and no new path is pushed to $Q$.

\begin{lemma}
	The above algorithm computes the $k$ shortest simple $s$-$t$ paths of a weighted, directed graph $G = (V, E)$.
\end{lemma}
\begin{proof}
	The algorithm uses the same idea of shortest deviations as existing $k$SSP algorithms or Eppstein's $k$SP algorithm.
	We only have to show that a non-simple path $p$ is processed before its simple enhancement $p'$, resulting from the suffix repair in the non-simple case, is actually needed.
	The set of nodes that are forbidden when the SP tree for $p$ is computed is a proper subset of the node set that the SP tree for $p'$ may not use.
	The suffix of $p$ is therefore not longer than that of $p'$, and $p$ is extracted from $Q$ (and subsequently, $p'$ is pushed) before we need to extract $p'$.
	\qed
\end{proof}

In terms of running time, the above algorithm requires too many computations of SP trees.

\begin{lemma}
	The running time of the above algorithm is $\mathcal O(km(m + n \log n))$.
\end{lemma}
\begin{proof}
	While processing a non-simple path, at most one new path is pushed to $Q$, which is always simple.
	Thus, the parent of a non-simple path is always simple.
    We have to process at most $k$ simple paths, each of which requires $\mathcal O(m + n)$ running time.
	Every simple path may have $\mathcal O(m)$ sidetracks extensions.
	In the worst case, all of them represent non-simple paths, yielding $\mathcal O(km)$ SP tree computations with a total running time of $\mathcal O(km(m + n\log n))$ if Dijkstra's algorithm with Fibonacci heaps is used.
	The running time for the non-simple cases clearly dominates.

	For every subset of $E$, there is at most one permutation of this subset that represents a simple $s$-$t$ path.
    The maximum number of paths enumerated by the algorithm is therefore $k' := \min\{k, 2^m\}$.
    We can limit the size of $Q$ efficiently to $k'$ using a double-ended priority queue~\cite{CUSTOM:book/sahni1999}.
    We push $\mathcal O(k'm)$ paths to $Q$ and extract $\mathcal O(k'm)$ paths from it; both operations require $\mathcal O(\log k')$ time on interval heaps.
    The total time spent on processing $Q$ is then $\mathcal O(k'm\log k') \subset \mathcal O(km^2)$.

    The pivot step requires $\mathcal O(n)$ running time for each of the $\mathcal O(k'm)$ extracted paths.
    \qed
\end{proof}

\section{Improvements}
\label{sec:improve}

We show how the number of SP tree computations can be reduced to $\mathcal O(kn)$ in the worst case, and seemingly even further in practice.

So far, we were only able to bound the number of SP tree computations by $\mathcal O(m)$ for each extracted simple path.
This stems from the fact that there may be $\mathcal O(m)$ sidetracks from such a path, each of them requiring a subsequent SP tree computation in the worst case.

Consider two sidetrack sequences $(e_1, \ldots, e_r, f_1 = (u, v))$, $(e_1, \ldots, e_r, f_2 = (u, w))$ that were added when a path $p$ represented by $(e_1, \ldots, e_r)$ was processed.
Let $p_1$, $p_2$ be the paths represented by these sequences, respectively.
Assume that both sequences represent non-simple paths, and therefore both require a new SP tree.
We assume w.l.o.g. that $p_1$ is extracted from $Q$ before $p_2$.

When $p_1$ is extracted from $Q$, we discover that it contains a cycle.
We then have to compute an SP tree $T$ for the graph $G - p'$, where $p'$ is the $s$-$u$ sub-path of $p$.
We push $(e_1, \ldots, e_r, f_1)$ back to $Q$, updating $T_{f_1} = T$.
When $p_2$ is extracted, the basic algorithm computes an SP tree for the exact same graph.
This computation may be skipped.
We check if an SP tree for this graph has already been computed, and reuse it if it exists.
In our case, we simply push $(e_1, \ldots, e_r, f_2)$ with $T_{f_2} = T$ to $Q$.

We obtain the following result.

\begin{lemma}
    Excluding the time spent on $Q$, the algorithm proposed in~\autoref{sec:simple-algo} in conjunction with SP tree reuse requires $\mathcal O(kn(m + n \log n))$ time to process non-simple paths.
\end{lemma}
\begin{proof}
    There are still $\mathcal O(km)$ many sequences in $Q$ that represent non-simple paths, but only $\mathcal O(kn)$ of them trigger an SP tree computation.
    Let $p$ be a non-simple path extracted from $Q$.
    The initial pivot step requires time $\mathcal O(n)$.
    If each path in $Q$ manages a pointer to its parent path as well as a pointer to the SP tree for $G - p'$ for every prefix path $p'$, already computed SP trees can be accessed in constant time.
    \qed
\end{proof}

The total running time of $\mathcal O(km^2)$ spent on $Q$ is now no longer dominated.
Instead of using a priority queue for the candidate paths, we organize all computed paths in a min-heap in the following way.
The shortest path is the root of the min-heap.
Whenever a path $p'$ is computed while a path $p$ is processed, we insert $p'$ into the min-heap as a child of $p$.
\autoref{fig:basic-example-heap} shows an example of such a min-heap.

We want to extract the $km$ smallest elements from this heap using Fredericksons heap selection algorithm~\cite{DBLP:journals/iandc/Frederickson93}.
The heap described above has maximum degree $m$, again yielding a running time of $\mathcal O(km^2)$.
Let $P_p$ be the set of paths found during the processing of $p$.
Instead of inserting every $p' \in P_p$ as a heap child of $p$, we heapify $P_p$ to obtain the heap $H_p$, using the lengths of the paths for keys again.
The root of $H_p$ is then inserted into the global min-heap as a child of $p$.
Note that the parent path of every path in $H_p$ is not its heap parent in $H_p$, but still $p$ itself.

Every simple path $p$ in the min-heap now has at most two heap successors with the same parent path as $p$, and at most one heap successor whose parent is $p$ itself.
Every non-simple path has at most one simple path as heap processor.
The maximum degree of the global min-heap is therefore bounded by three and Frederickson's heap selection can be done in time $\mathcal O(km)$.

\begin{corollary}
    The algorithm proposed in~\autoref{sec:simple-algo} in conjunction with SP tree reuse and Frederickson's heap selection algorithm computes the $k$ shortest simple $s$-$t$ paths of a weighted, directed graph $G = (V, E)$, $s, t \in V$, in $\mathcal O(kn(m + n \log n))$ time.
\end{corollary}

We propose two more modifications that do not change the asymptotic worst-case running time.
We provide evidence that these changes make the algorithm faster in practice.

Consider one of the $k$ simple $s$-$t$ paths $p$ represented by sidetracks $(e_1, \ldots, e_r)$ with $e_i = (v_{i - 1}, v_i)$, $s = v_0$ and $t = v_r$.
When $p$ is processed, we push the set $P_p$ of paths to $Q$, with $\norm{P_p} \in \mathcal O(m)$.
The basic algorithm tests for each $p' \in P_p$ if $p'$ is simple in time $\mathcal O(n)$, leading to a total time of $\mathcal O(kmn)$ for these tests.

Let $T = T_{e_r}$.
By removing all $e_i$ from $T$, the SP tree decomposes into a set of trees $T_i$ such that $T_i$ is rooted in $v_i$.
The \emph{block $i$} is the node set of $T_i$.
Observe that the path $p'$ represented by a sequence $(e_1, \ldots, e_r, e)$, $e = (v_i, w)$, with $v_i$, $w$ in block $i$, $j$, respectively, is simple iff $i < j$.
If $i \ge j$, we follow $p$ until we reach $v_i$, traverse $e$ and follow $T$ to reach $v_i$ again.
Otherwise, the first node on $p$ we hit after deviating from it via $e$ is $v_j$.
Since $i < j$, the $v_j$-$t$ subpath of $p$ does not contain $v_i$, and therefore, $p'$ is simple.
The partition of $V$ into blocks can be computed in time $\mathcal O(n)$.
We can then collect all sidetracks deviating from $p$ and check for each of them if their heads belong to a smaller block than their tails in $\mathcal O(m)$ total time.
We store this information along with the corresponding sidetrack sequences in $Q$.
The pivot turn is replaced by a constant time lookup.
All tests for simplicity then require time $\mathcal O(k(m + n))$ instead of $\mathcal O(kmn)$.

Finally, we want to reduce the number of SP tree computations in practice.
Let $p$ be a non-simple path represented by the sequence $(e_1, \ldots, e_r)$ with $e_i = (v_{i - 1}, v_i)$.
After we discover that $p$ is not simple, the basic algorithm computes an SP tree in $G - \pref(p)$.
Only then does the algorithm check if $v_r \in T$.

Obviously, there is a shortest $w$-$t$ path in $G - \pref(p)$, i.e., $w \in T$, iff there is \emph{some} directed path from $w$ to $t$ in $G - \pref(p)$.
Latter can be checked by a much simpler reachability check in time $\mathcal O(m + n)$.
A naive approach checks reachability for every combination of some node $v \in V$ and one of the $\mathcal O(kn)$ nodes on the output simple paths, yielding time $\mathcal O(kn^2(m + n))$.
Of course, for a fixed prefix $p^i$ of some simple path $p$, we can also check reachability in $G - p^i$ in $\mathcal O(m + n)$ time for every node in $G$ simultaneously, and obtain $\mathcal O(kn(m + n))$ total time.

Let $l$ be the number of edges in $p$, and consider sidetracks $(v_{l - 1}, w)$.
To determine whether the SP path computation for this sidetrack is necessary, we have to check whether there is a path in $G^{l - 1}$, where $l$ is the number of edges of $p$, and $G^i := G - p^i$.
We determine reachability in $G^{l - 1}$ by starting a reverse depth-first search from $t$, ignoring every node that lies on $p$.
After this search, $w$-$t$ reachability can be evaluated in constant time per sidetrack $(v_{l - 1}, w)$.
If $w$ turns out to be separated from $t$, the path represented by $(e_1, \ldots, e_r, (v_{l - 1}, w))$ is non-simple and cannot be repaired to a simple path.
In this case, we discard the sidetrack, which in turn cannot trigger an SP tree computation as it is never extracted from $Q$.

After collecting sidetracks emanating from $v_{i + 1}$ on $p$, we continue with sidetracks emanating from the predecessor $v_i$.
We conduct a depth-first search again, this time starting in $v_i$ and reusing the reachability information computed before.
This way, we only process nodes that were unreachable before.
In other words, we solve an incremental series of reachability instances.
This procedure terminates when $v_i$ is the deviation node of $p$, and takes total time $\mathcal O(k(m + n))$.

\section{Experiments}
\label{sec:experiments}

\begin{table}[tb]
\scriptsize
\centering
\begin{tabularx}{1\textwidth}{p{0.8cm}p{0.8cm}|RR|RR|RR|RR|RR}
 & & \multicolumn{2}{c|}{$m=2n$} & \multicolumn{2}{c|}{$m=4n$} & \multicolumn{2}{c|}{$m=10n$} & \multicolumn{2}{c|}{$m=30n$} & \multicolumn{2}{c}{$m=50n$} \\
 \centering $n$ & & Med & $Q_{.9}$ & Med & $Q_{.9}$ & Med & $Q_{.9}$ & Med & $Q_{.9}$ & Med & $Q_{.9}$ \\
\hline
\centering\multirow{3}{*}{2,000} & \centering NC & 0.91 & 2.24 & 0.41 & 1.06 & 0.35 & 1.14 & 0.41 & 1.39 & 1.95 & 3.41 \\
& \centering SB-r  & 0.18 & 0.23 & 0.27 & 0.29 & 0.72 & 0.75 & 3.70 & 3.86 & 8.94 & 9.08 \\
& \centering SB-o  & \textbf{0.09} & \textbf{0.17} & \textbf{0.07} & \textbf{0.09} & \textbf{0.09} & \textbf{0.12} & \textbf{0.11} & \textbf{0.16} & \textbf{0.16} & \textbf{0.29} \\
\hline
\centering\multirow{3}{*}{4,000} & \centering NC & 0.90 & 2.63 & 0.76 & 2.39 & 0.75 & 1.79 & 1.24 & 4.61 & 1.92 & 4.92 \\
& \centering SB-r  & 0.35 & 0.40 & 0.53 & 0.58 & 1.39 & 1.60 & 7.26 & 7.31 & 17.46 & 17.64 \\
& \centering SB-o  & \textbf{0.16} & \textbf{0.21} & \textbf{0.12} & \textbf{0.17} & \textbf{0.13} & \textbf{0.20} & \textbf{0.17} & \textbf{0.22} & \textbf{0.27} & \textbf{0.38} \\
\hline
\centering\multirow{3}{*}{6,000} & \centering NC & 2.99 & 5.88 & 0.47 & 1.53 & 0.65 & 3.06 & 1.93 & 7.14 & 2.07 & 8.37 \\
& \centering SB-r  & 0.54 & 0.61 & 0.81 & 0.83 & 2.11 & 2.18 & 10.95 & 11.08 & 26.53 & 26.70 \\
& \centering SB-o  & \textbf{0.23} & \textbf{0.30} & \textbf{0.21} & \textbf{0.26} & \textbf{0.20} & \textbf{0.27} & \textbf{0.27} & \textbf{0.37} & \textbf{0.31} & \textbf{0.50} \\
\hline
\centering\multirow{3}{*}{8,000} & \centering NC & 1.62 & 7.16 & 0.62 & 2.52 & 1.79 & 4.27 & 3.29 & 9.23 & 2.45 & 8.94 \\
& \centering SB-r  & 0.69 & 0.86 & 1.08 & 1.11 & 2.81 & 2.89 & 14.66 & 15.35 & 34.84 & 35.50 \\
& \centering SB-o  & \textbf{0.28} & \textbf{0.46} & \textbf{0.23} & \textbf{0.26} & \textbf{0.32} & \textbf{0.41} & \textbf{0.34} & \textbf{0.54} & \textbf{0.39} & \textbf{0.47} \\
\hline
\centering\multirow{3}{*}{10,000} & \centering NC & 1.70 & 10.86 & 1.09 & 5.00 & 2.46 & 8.77 & 5.83 & 12.19 & 8.96 & 23.97 \\
& \centering SB-r  & 0.87 & 0.92 & 1.37 & 1.48 & 3.63 & 3.81 & 18.42 & 18.64 & 43.59 & 43.92 \\
& \centering SB-o  & \textbf{0.35} & \textbf{0.40} & \textbf{0.30} & \textbf{0.42} & \textbf{0.33} & \textbf{0.38} & \textbf{0.37} & \textbf{0.44} & \textbf{0.45} & \textbf{0.61} \\
\end{tabularx}
\vspace{2mm}
\caption{Median and 90\% quantile $Q_{.9}$ of running times in seconds on random graphs with $k = 2000$.}
\label{table:random}
\end{table}

To demonstrate the effectiveness of our algorithm, we conducted a series of experiments.
Feng~\cite{DBLP:journals/networks/Feng14} showed recently that their algorithm is the most efficient one in practice.
We therefore only compare our algorithm to Feng's node classification algorithm (NC).
Our implementation of NC does not use express edges and achieves better running times than the implementation of Feng, who appears to have used the same processor as we did.
We implemented two variants of our algorithm, both of which determine simplicity of all sidetracks of a simple path at once by partitioning the nodes into blocks as discussed above.
The first version, SB-r, tries to reduce the number of SP tree computations by solving some reachability problems; the second version, SB-o, spares this measure and is thus more optimistic.
None of the implementations uses Frederickson's heap selection algorithm, resulting in an additional running time $\mathcal O(km\log k)$ for SB-r and SB-o, but not for NC.
For space restrictions, we contented ourselves with two graph classes that Feng used in their experiments, including road graphs that are especially relevant in practice.

We implemented all algorithms in C++, using forward and reverse star representation for directed graphs.
Shortest paths (NC) and SP trees (SB-r, SB-o) are computed using a common implementation of Dijkstra's algorithm; tentative labels are managed by a pairing heap.
Our implementation of Dijkstra's algorithm stops as soon as the label of $t$ is made permanent if only a single pair shortest path is needed, which is essential for NC.
The queue of candidate paths $Q$ is implemented as an interval heap, a form of double-headed priority queues, which allows us to limit its size efficiently to the number of simple paths that have yet to be output.
Special care has to be taken here for SB-r and SB-o since $Q$ also has to manage non-simple paths.
The experiments ran on an Intel Core i7-3770 @ 3.40GHz with 16GB of RAM on a GNU/Gentoo Linux with kernel version \texttt{4.2.5} and TurboBoost turned off.
Source code was compiled using the GNU C++ compiler \texttt{g++-4.9.3} and \texttt{-O3} optimization.

\subsection{Random Graphs}
\label{sec:random-graphs}

\begin{table}[tb]
\scriptsize
\centering
\begin{tabularx}{1\textwidth}{p{0.8cm}p{0.8cm}p{0.8cm}|RR|RR|RR}
 & & & \multicolumn{2}{c|}{$m=2n$} & \multicolumn{2}{c|}{$m=10n$} & \multicolumn{2}{c}{$m=50n$} \\
 \centering $k$ & \centering $n$ & & Dijkstra & Polls & Dijkstra & Polls & Dijkstra & Polls \\
\hline
\centering\multirow{15}{*}{$\num{2000}$} & \centering\multirow{3}{*}{$\num{2000}$} & \centering NC & $\num{22981}$ & $\num{4402186}$ & $\num{14532}$ & $\num{613888}$ & $\num{14512}$ & $\num{2105908}$ \\
 & & \centering SB-r  & $\num{80}$ & $\num{158740}$ & $\num{65}$ & $\num{129426}$ & $\num{44}$ & $\num{87640}$ \\
 & & \centering SB-o  & $\num{82}$ & $\num{162774}$ & $\num{65}$ & $\num{129426}$ & $\num{44}$ & $\num{87640}$ \\
\cline{2-9}
 & \centering\multirow{3}{*}{$\num{4000}$} & \centering NC & $\num{25130}$ & $\num{3284822}$ & $\num{14580}$ & $\num{1465676}$ & $\num{15604}$ & $\num{1404670}$ \\
 & & \centering SB-r  & $\num{44}$ & $\num{177236}$ & $\num{20}$ & $\num{77830}$ & $\num{28}$ & $\num{113790}$ \\
 & & \centering SB-o  & $\num{46}$ & $\num{185205}$ & $\num{20}$ & $\num{77830}$ & $\num{28}$ & $\num{113790}$ \\
\cline{2-9}
 & \centering\multirow{3}{*}{$\num{6000}$} & \centering NC & $\num{26990}$ & $\num{12142440}$ & $\num{16652}$ & $\num{719338}$ & $\num{16444}$ & $\num{1151469}$ \\
 & & \centering SB-r  & $\num{43}$ & $\num{257184}$ & $\num{18}$ & $\num{110849}$ & $\num{21}$ & $\num{125826}$ \\
 & & \centering SB-o  & $\num{44}$ & $\num{263158}$ & $\num{18}$ & $\num{110849}$ & $\num{21}$ & $\num{125826}$ \\
\cline{2-9}
 & \centering\multirow{3}{*}{$\num{8000}$} & \centering NC & $\num{26810}$ & $\num{5478902}$ & $\num{17316}$ & $\num{2418042}$ & $\num{17034}$ & $\num{1109572}$ \\
 & & \centering SB-r  & $\num{25}$ & $\num{199528}$ & $\num{16}$ & $\num{131850}$ & $\num{17}$ & $\num{135853}$ \\
 & & \centering SB-o  & $\num{28}$ & $\num{223446}$ & $\num{16}$ & $\num{131850}$ & $\num{17}$ & $\num{135853}$ \\
\cline{2-9}
 & \centering\multirow{3}{*}{$\num{10000}$} & \centering NC & $\num{26633}$ & $\num{4806279}$ & $\num{17826}$ & $\num{3552732}$ & $\num{18186}$ & $\num{6130731}$ \\
 & & \centering SB-r  & $\num{23}$ & $\num{229591}$ & $\num{15}$ & $\num{149870}$ & $\num{14}$ & $\num{134892}$ \\
 & & \centering SB-o  & $\num{23}$ & $\num{229634}$ & $\num{15}$ & $\num{149870}$ & $\num{14}$ & $\num{134892}$ \\
\end{tabularx}
\vspace{2mm}
\caption{Median number of Dijkstra calls and polls for random graphs.}
\label{table:dijkstras}
\end{table}

We first considered random graphs generated by the \texttt{sprand} generator provided on the website of the ninth DIMACS implementation challenge~\cite{CUSTUM:dimacs9}.
The generator draws at random a fixed amount of edges, possibly resulting in a multigraph.
For each combination of graph size $n \in \{\num{2000}, \num{4000}, \num{6000}, \num{8000}, \num{10000}\}$ and \emph{linear density} $m / n \in \{2, 3, 4, 7, 10, 20, 30, 40, 50\}$, we generated 20 random graphs, and enumerated $k \in \{\num{200}, \num{500}, \num{1000}, \num{2000}\}$ simple paths.

In \autoref{table:random}, the median and 90\% quantile $Q_{.9}$ (90\% of the running times were at most $Q_{.9}$) of execution times for some densities are summarized.
For small densities, we observe that SB-r is faster than NC, but becomes much slower as the density grows.
The running time of SB-r increases by a factor of 50 between $m = 2n$ and $m = 50n$, but only by a factor of 2 for NC.
SB-o is about twice as fast as SB-r for very small densities, and is even more robust against density changes than NC.
SB-o is thus the fastest of the three algorithms for all graph sizes and densities.
Also note the very low dispersion of SB running times.
For NC, the 90\% quantile of the running time is regularly three times the median running time, and even exceeds a factor of 6 for $n = \num{10000}$ and $m = \num{20000}$.
In contrast, this quotient is always much closer to 1 for SB-o and assumes its maximum of 1.87 for $n = \num{2000}$, $m = \num{100000}$.
We can therefore predict the running time of SB-o much more accurate than that of NC.
The $Q_{.9}$ running time of SB-o is still well below the median running time of both SB-r and NC.
The dispersion of SB-r is even lower.

\begin{table}[tb]
\centering
\begin{subtable}[b]{1\textwidth}
\scriptsize
\centering
\begin{tabularx}{0.7\textwidth}{XRRRR}
Graph & NY & BAY & COL & FLA \\
\hline
$n$ & $\num{264345}$ & $\num{321270}$ & $\num{345666}$ & $\num{1070376}$ \\
$m$ & $\num{733846}$ & $\num{800172}$ & $\num{1057066}$ & $\num{2712798}$ \\
\end{tabularx}
\caption{Sizes of four road graphs.}
\label{table:tiger-sizes}
\end{subtable}
\begin{subtable}[b]{1\textwidth}
\scriptsize
\centering
\begin{tabularx}{1\textwidth}{p{0.8cm}p{0.8cm}|RR>{\RaggedLeft}p{1.4cm}|RR>{\RaggedLeft}p{1.5cm}|RR>{\RaggedLeft}p{1.5cm}}
 & & \multicolumn{3}{c|}{$k=100$} & \multicolumn{3}{c|}{$k=200$} & \multicolumn{3}{c}{$k=300$} \\
 & & Med & $Q_{.9}$ & Polls & Med & $Q_{.9}$ & Polls & Med & $Q_{.9}$ & Polls \\
\hline
\centering\multirow{3}{*}{NY} & \centering NC & 2.06 & 12.14 & $\num{3918630}$ & 3.77 & 24.11 & $\num{6969812}$ & 5.40 & 35.17 & $\num{9803818}$ \\
 & \centering SB-r  & 1.38 & 3.60 & $\num{528614}$ & 2.64 & 7.62 & $\num{792628}$ & 3.88 & 11.43 & $\num{1055890}$ \\
 & \centering SB-o  & \textbf{0.55} & \textbf{2.77} & $\num{528614}$ & \textbf{0.97} & \textbf{5.91} & $\num{792628}$ & \textbf{1.38} & \textbf{9.10} & $\num{1055890}$ \\
\hline
\centering\multirow{3}{*}{BAY} & \centering NC & 5.11 & 17.09 & $\num{15215868}$ & 9.27 & 33.81 & $\num{28784851}$ & 13.84 & 49.54 & $\num{38731878}$ \\
 & \centering SB-r  & 1.76 & 8.43 & $\num{963245}$ & 3.15 & 18.52 & $\num{1761843}$ & 4.77 & 28.18 & $\num{1922360}$ \\
 & \centering SB-o  & \textbf{0.79} & \textbf{7.62} & $\num{963245}$ & \textbf{1.49} & \textbf{17.67} & $\num{1761843}$ & \textbf{1.99} & \textbf{24.97} & $\num{1922360}$ \\
\hline
\centering\multirow{3}{*}{COL} & \centering NC & 6.53 & 25.13 & $\num{16459836}$ & 11.65 & 44.08 & $\num{30058602}$ & 15.98 & \textbf{58.83} & $\num{42430752}$ \\
 & \centering SB-r  & 2.10 & \textbf{18.06} & $\num{435666}$ & 4.01 & 38.40 & $\num{435666}$ & 6.00 & 62.28 & $\num{435666}$ \\
 & \centering SB-o  & \textbf{0.80} & 18.43 & $\num{435666}$ & \textbf{1.42} & \textbf{37.98} & $\num{435666}$ & \textbf{2.02} & 60.62 & $\num{435666}$ \\
\hline
\centering\multirow{3}{*}{FLA} & \centering NC & 30.15 & 67.43 & $\num{56950588}$ & 58.13 & 126.24 & $\num{107818950}$ & 83.00 & 188.03 & $\num{151950959}$ \\
 & \centering SB-r  & 5.53 & 9.68 & $\num{1070376}$ & 10.60 & 33.13 & $\num{1070376}$ & 15.73 & 55.29 & $\num{1070376}$ \\
 & \centering SB-o  & \textbf{2.54} & \textbf{6.81} & $\num{1070376}$ & \textbf{4.65} & \textbf{27.72} & $\num{1070376}$ & \textbf{6.78} & \textbf{47.00} & $\num{1070376}$ \\
\end{tabularx}
\caption{Median and 90\% quantile $Q_{.9}$ of running times in seconds, median number of polls.}
\label{table:tiger-times}
\end{subtable}
\vspace{-4mm}
\caption{Sizes and metrics for four large TIGER road graphs.}
\label{fig:tiger}
\end{table}

\autoref{table:dijkstras} shows the median number of Dijkstra calls.
The numbers are relatively stable across the various densities, but the Dijkstra counts for the SB algorithms is orders of magnitudes smaller than the count for the NC algorithm.
Note, however, that SB needs to compute the complete SP tree every time.
In contrast, NC only solves single pair shortest path problems on rather small subgraphs.
We also provide the number of polls, i.e.,  the total number of nodes that were extracted from Dijkstra's priority queue, for comparability.
NC still requires an order of magnitude more polls compared to SB.
Further, there are only small differences in the median number of polls for SB-r and SB-o for $m = 2n$.
Solving incremental reach for each of the $\num{2000}$ output simple paths only reduced the number of SP tree computations by at most three for $n = \num{8000}$, $m = 2n$.
The extra effort involved to reduce SP tree computations ranges from 50\% of the total running time on very sparse graphs to 98\% on very dense graphs, and does not pay off.
Therefore, SB-o is clearly the fastest algorithm on random graphs.

\subsection{Road Graphs}
\label{sec:road-graphs}

We considered road graphs of various areas in the USA called TIGER graphs, again provided by the DIMACS website~\cite{CUSTUM:dimacs9}.
In particular, we used the road networks of New York (NY), the San Francisco Bay Area (BAY), Colorado (COL), and Florida (FLA).
The sizes of these graphs are shown in \autoref{table:tiger-sizes}.
We drew 20 $s$-$t$ pairs at random and enumerated $k \in \{100, 200, 300\}$ paths.

The resulting running times are summarized in \autoref{table:tiger-times}, along with the median number of polls.
The median running time of NC is clearly dominated by both SB variants.
SB-o achieves a minimum speedup around 4 on NY across all values of $k$; on FLA, the speedup is roughly 12.
SB-r takes approximately twice the time of SB-o and is still much faster than NC.
The ratio of $Q_{.9}$ and median running time is worse for SB, but the $Q_{.9}$ time itself of SB-o is still better than that of NC except for 300 paths on COL.
On FLA, the largest graph, the 90\% quantile of SB-o is much better than the median running time of NC.

\bibliography{literature}
\bibliographystyle{splncs03}

\appendix



\begin{sidewaystable}
\scriptsize
\centering
\begin{tabularx}{1\textwidth}{p{0.8cm}p{0.8cm}|RR|RR|RR|RR|RR|RR|RR|RR|RR}
  & & \multicolumn{18}{c}{$k = 500$} \\
  & & \multicolumn{2}{c|}{$m=2n$} & \multicolumn{2}{c|}{$m=3n$} & \multicolumn{2}{c|}{$m=4n$} & \multicolumn{2}{c|}{$m=7n$} & \multicolumn{2}{c}{$m=10n$} & \multicolumn{2}{c}{$m=20n$} & \multicolumn{2}{c}{$m=30n$} & \multicolumn{2}{c}{$m=40n$} & \multicolumn{2}{c}{$m=50n$} \\
\centering $n$ & & Med & $Q_{.9}$ & Med & $Q_{.9}$ & Med & $Q_{.9}$ & Med & $Q_{.9}$ & Med & $Q_{.9}$ & Med & $Q_{.9}$ & Med & $Q_{.9}$ & Med & $Q_{.9}$ & Med & $Q_{.9}$ \\
\hline
\centering\multirow{3}{*}{$\num{2000}$} & \centering NC & 0.22 & 0.56 & 0.14 & 0.29 & 0.10 & 0.27 & 0.06 & 0.20 & 0.08 & 0.27 & 0.11 & 0.39 & 0.10 & 0.35 & 0.23 & 0.54 & 0.47 & 0.84 \\
 & \centering SB-r  & 0.04 & 0.06 & 0.05 & 0.07 & 0.07 & 0.07 & 0.11 & 0.13 & 0.18 & 0.19 & 0.49 & 0.51 & 0.92 & 0.95 & 1.50 & 1.53 & 2.22 & 2.26 \\
 & \centering SB-o  & \textbf{0.02} & \textbf{0.04} & \textbf{0.02} & \textbf{0.03} & \textbf{0.02} & \textbf{0.02} & \textbf{0.02} & \textbf{0.03} & \textbf{0.02} & \textbf{0.03} & \textbf{0.03} & \textbf{0.05} & \textbf{0.03} & \textbf{0.04} & \textbf{0.03} & \textbf{0.04} & \textbf{0.04} & \textbf{0.07} \\
\hline
\centering\multirow{3}{*}{$\num{4000}$} & \centering NC & 0.22 & 0.65 & 0.15 & 0.51 & 0.19 & 0.60 & 0.20 & 0.55 & 0.18 & 0.43 & 0.71 & 1.21 & 0.31 & 1.14 & 0.29 & 1.24 & 0.48 & 1.27 \\
 & \centering SB-r  & 0.09 & 0.10 & 0.11 & 0.13 & 0.13 & 0.15 & 0.24 & 0.26 & 0.35 & 0.38 & 0.96 & 0.99 & 1.81 & 1.83 & 2.96 & 3.60 & 4.34 & 4.74 \\
 & \centering SB-o  & \textbf{0.04} & \textbf{0.04} & \textbf{0.03} & \textbf{0.05} & \textbf{0.03} & \textbf{0.05} & \textbf{0.04} & \textbf{0.06} & \textbf{0.03} & \textbf{0.05} & \textbf{0.04} & \textbf{0.06} & \textbf{0.04} & \textbf{0.06} & \textbf{0.05} & \textbf{0.08} & \textbf{0.05} & \textbf{0.08} \\
\hline
\centering\multirow{3}{*}{$\num{6000}$} & \centering NC & 0.74 & 1.47 & 0.13 & 0.95 & 0.12 & 0.39 & 0.30 & 0.75 & 0.17 & 0.78 & 0.28 & 1.87 & 0.48 & 1.78 & 0.39 & 1.86 & 0.51 & 2.09 \\
 & \centering SB-r  & 0.14 & 0.15 & 0.16 & 0.17 & 0.20 & 0.21 & 0.34 & 0.36 & 0.52 & 0.54 & 1.44 & 1.56 & 2.74 & 2.76 & 4.40 & 4.44 & 6.56 & 6.59 \\
 & \centering SB-o  & \textbf{0.06} & \textbf{0.07} & \textbf{0.05} & \textbf{0.06} & \textbf{0.05} & \textbf{0.06} & \textbf{0.05} & \textbf{0.06} & \textbf{0.05} & \textbf{0.07} & \textbf{0.06} & \textbf{0.11} & \textbf{0.07} & \textbf{0.09} & \textbf{0.07} & \textbf{0.09} & \textbf{0.07} & \textbf{0.11} \\
\hline
\centering\multirow{3}{*}{$\num{8000}$} & \centering NC & 0.40 & 1.78 & 0.51 & 1.21 & 0.15 & 0.62 & 0.27 & 1.53 & 0.44 & 1.06 & 0.36 & 2.06 & 0.82 & 2.20 & 0.76 & 2.56 & 0.62 & 2.46 \\
 & \centering SB-r  & 0.17 & 0.20 & 0.22 & 0.24 & 0.27 & 0.28 & 0.46 & 0.47 & 0.71 & 0.75 & 1.94 & 2.18 & 3.64 & 3.72 & 5.93 & 6.21 & 8.72 & 9.69 \\
 & \centering SB-o  & \textbf{0.07} & \textbf{0.10} & \textbf{0.06} & \textbf{0.09} & \textbf{0.06} & \textbf{0.07} & \textbf{0.06} & \textbf{0.07} & \textbf{0.07} & \textbf{0.09} & \textbf{0.07} & \textbf{0.09} & \textbf{0.08} & \textbf{0.12} & \textbf{0.09} & \textbf{0.13} & \textbf{0.09} & \textbf{0.11} \\
\hline
\centering\multirow{3}{*}{$\num{10000}$} & \centering NC & 0.42 & 2.73 & 0.36 & 1.97 & 0.28 & 1.24 & 0.38 & 2.25 & 0.61 & 2.09 & 0.80 & 2.35 & 1.45 & 3.01 & 1.48 & 3.81 & 2.25 & 5.89 \\
 & \centering SB-r  & 0.21 & 0.23 & 0.27 & 0.31 & 0.34 & 0.38 & 0.58 & 0.60 & 0.91 & 0.99 & 2.44 & 2.67 & 4.61 & 4.62 & 7.42 & 7.59 & 10.88 & 10.93 \\
 & \centering SB-o  & \textbf{0.08} & \textbf{0.10} & \textbf{0.08} & \textbf{0.11} & \textbf{0.07} & \textbf{0.11} & \textbf{0.08} & \textbf{0.09} & \textbf{0.08} & \textbf{0.10} & \textbf{0.10} & \textbf{0.13} & \textbf{0.09} & \textbf{0.12} & \textbf{0.10} & \textbf{0.13} & \textbf{0.11} & \textbf{0.15} \\
\hline
\\
  & & \multicolumn{18}{c}{$k = 2000$} \\
  & & \multicolumn{2}{c|}{$m=2n$} & \multicolumn{2}{c|}{$m=3n$} & \multicolumn{2}{c|}{$m=4n$} & \multicolumn{2}{c|}{$m=7n$} & \multicolumn{2}{c}{$m=10n$} & \multicolumn{2}{c}{$m=20n$} & \multicolumn{2}{c}{$m=30n$} & \multicolumn{2}{c}{$m=40n$} & \multicolumn{2}{c}{$m=50n$} \\
\centering $n$ & & Med & $Q_{.9}$ & Med & $Q_{.9}$ & Med & $Q_{.9}$ & Med & $Q_{.9}$ & Med & $Q_{.9}$ & Med & $Q_{.9}$ & Med & $Q_{.9}$ & Med & $Q_{.9}$ & Med & $Q_{.9}$ \\
\hline
\centering\multirow{3}{*}{$\num{2000}$} & \centering NC & 0.91 & 2.24 & 0.58 & 1.18 & 0.41 & 1.06 & 0.30 & 0.81 & 0.35 & 1.14 & 0.46 & 1.54 & 0.41 & 1.39 & 0.99 & 2.22 & 1.95 & 3.41 \\
 & \centering SB-r  & 0.18 & 0.23 & 0.22 & 0.26 & 0.27 & 0.29 & 0.46 & 0.50 & 0.72 & 0.75 & 1.97 & 2.04 & 3.70 & 3.86 & 6.10 & 6.32 & 8.94 & 9.08 \\
 & \centering SB-o  & \textbf{0.09} & \textbf{0.17} & \textbf{0.07} & \textbf{0.11} & \textbf{0.07} & \textbf{0.09} & \textbf{0.06} & \textbf{0.11} & \textbf{0.09} & \textbf{0.12} & \textbf{0.10} & \textbf{0.17} & \textbf{0.11} & \textbf{0.16} & \textbf{0.14} & \textbf{0.20} & \textbf{0.16} & \textbf{0.29} \\
\hline
\centering\multirow{3}{*}{$\num{4000}$} & \centering NC & 0.90 & 2.63 & 0.58 & 2.04 & 0.76 & 2.39 & 0.72 & 1.78 & 0.75 & 1.79 & 2.88 & 4.78 & 1.24 & 4.61 & 1.19 & 5.00 & 1.92 & 4.92 \\
 & \centering SB-r  & 0.35 & 0.40 & 0.44 & 0.49 & 0.53 & 0.58 & 0.95 & 1.02 & 1.39 & 1.60 & 3.85 & 3.95 & 7.26 & 7.31 & 11.87 & 12.56 & 17.46 & 17.64 \\
 & \centering SB-o  & \textbf{0.16} & \textbf{0.21} & \textbf{0.13} & \textbf{0.19} & \textbf{0.12} & \textbf{0.17} & \textbf{0.16} & \textbf{0.23} & \textbf{0.13} & \textbf{0.20} & \textbf{0.17} & \textbf{0.26} & \textbf{0.17} & \textbf{0.22} & \textbf{0.21} & \textbf{0.34} & \textbf{0.27} & \textbf{0.38} \\
\hline
\centering\multirow{3}{*}{$\num{6000}$} & \centering NC & 2.99 & 5.88 & 0.49 & 3.78 & 0.47 & 1.53 & 1.20 & 2.93 & 0.65 & 3.06 & 1.18 & 7.44 & 1.93 & 7.14 & 1.58 & 7.28 & 2.07 & 8.37 \\
 & \centering SB-r  & 0.54 & 0.61 & 0.66 & 0.68 & 0.81 & 0.83 & 1.36 & 1.45 & 2.11 & 2.18 & 5.79 & 6.46 & 10.95 & 11.08 & 17.62 & 17.84 & 26.53 & 26.70 \\
 & \centering SB-o  & \textbf{0.23} & \textbf{0.30} & \textbf{0.20} & \textbf{0.22} & \textbf{0.21} & \textbf{0.26} & \textbf{0.22} & \textbf{0.27} & \textbf{0.20} & \textbf{0.27} & \textbf{0.22} & \textbf{0.35} & \textbf{0.27} & \textbf{0.37} & \textbf{0.26} & \textbf{0.32} & \textbf{0.31} & \textbf{0.50} \\
\hline
\centering\multirow{3}{*}{$\num{8000}$} & \centering NC & 1.62 & 7.16 & 1.98 & 4.97 & 0.62 & 2.52 & 1.05 & 6.24 & 1.79 & 4.27 & 1.40 & 8.22 & 3.29 & 9.23 & 3.04 & 10.36 & 2.45 & 8.94 \\
 & \centering SB-r  & 0.69 & 0.86 & 0.88 & 0.96 & 1.08 & 1.11 & 1.85 & 1.87 & 2.81 & 2.89 & 7.77 & 8.57 & 14.66 & 15.35 & 23.78 & 24.95 & 34.84 & 35.50 \\
 & \centering SB-o  & \textbf{0.28} & \textbf{0.46} & \textbf{0.26} & \textbf{0.34} & \textbf{0.23} & \textbf{0.26} & \textbf{0.24} & \textbf{0.27} & \textbf{0.32} & \textbf{0.41} & \textbf{0.30} & \textbf{0.34} & \textbf{0.34} & \textbf{0.54} & \textbf{0.38} & \textbf{0.47} & \textbf{0.39} & \textbf{0.47} \\
\hline
\centering\multirow{3}{*}{$\num{10000}$} & \centering NC & 1.70 & 10.86 & 1.56 & 7.75 & 1.09 & 5.00 & 1.56 & 9.06 & 2.46 & 8.77 & 3.18 & 9.55 & 5.83 & 12.19 & 5.92 & 15.18 & 8.96 & 23.97 \\
 & \centering SB-r  & 0.87 & 0.92 & 1.09 & 1.25 & 1.37 & 1.48 & 2.33 & 2.39 & 3.63 & 3.81 & 9.83 & 10.28 & 18.42 & 18.64 & 29.76 & 29.86 & 43.59 & 43.92 \\
 & \centering SB-o  & \textbf{0.35} & \textbf{0.40} & \textbf{0.30} & \textbf{0.47} & \textbf{0.30} & \textbf{0.42} & \textbf{0.30} & \textbf{0.37} & \textbf{0.33} & \textbf{0.38} & \textbf{0.37} & \textbf{0.47} & \textbf{0.37} & \textbf{0.44} & \textbf{0.41} & \textbf{0.49} & \textbf{0.45} & \textbf{0.61} \\
\hline
\\
  & & \multicolumn{18}{c}{$k = 2000$} \\
  & & \multicolumn{2}{c|}{$m=2n$} & \multicolumn{2}{c|}{$m=3n$} & \multicolumn{2}{c|}{$m=4n$} & \multicolumn{2}{c|}{$m=7n$} & \multicolumn{2}{c}{$m=10n$} & \multicolumn{2}{c}{$m=20n$} & \multicolumn{2}{c}{$m=30n$} & \multicolumn{2}{c}{$m=40n$} & \multicolumn{2}{c}{$m=50n$} \\
\centering $n$ & & Med & $Q_{.9}$ & Med & $Q_{.9}$ & Med & $Q_{.9}$ & Med & $Q_{.9}$ & Med & $Q_{.9}$ & Med & $Q_{.9}$ & Med & $Q_{.9}$ & Med & $Q_{.9}$ & Med & $Q_{.9}$ \\
\hline
\centering\multirow{2}{*}{$\num{2000}$} & \centering NC & 0.91 & 2.24 & 0.58 & 1.18 & 0.41 & 1.06 & 0.30 & 0.81 & 0.35 & 1.14 & 0.46 & 1.54 & 0.41 & 1.39 & 0.99 & 2.22 & 1.95 & 3.41 \\
 & \centering SB-o  & \textbf{0.07} & \textbf{0.15} & \textbf{0.04} & \textbf{0.09} & \textbf{0.05} & \textbf{0.07} & \textbf{0.05} & \textbf{0.09} & \textbf{0.08} & \textbf{0.12} & \textbf{0.10} & \textbf{0.16} & \textbf{0.11} & \textbf{0.16} & \textbf{0.15} & \textbf{0.20} & \textbf{0.17} & \textbf{0.29} \\
\hline
\centering\multirow{2}{*}{$\num{4000}$} & \centering NC & 0.90 & 2.63 & 0.58 & 2.04 & 0.76 & 2.39 & 0.72 & 1.78 & 0.75 & 1.79 & 2.88 & 4.78 & 1.24 & 4.61 & 1.19 & 5.00 & 1.92 & 4.92 \\
 & \centering SB-o  & \textbf{0.09} & \textbf{0.13} & \textbf{0.07} & \textbf{0.12} & \textbf{0.06} & \textbf{0.10} & \textbf{0.10} & \textbf{0.18} & \textbf{0.08} & \textbf{0.14} & \textbf{0.14} & \textbf{0.21} & \textbf{0.15} & \textbf{0.18} & \textbf{0.19} & \textbf{0.28} & \textbf{0.22} & \textbf{0.32} \\
\hline
\centering\multirow{2}{*}{$\num{6000}$} & \centering NC & 2.99 & 5.88 & 0.49 & 3.78 & 0.47 & 1.53 & 1.20 & 2.93 & 0.65 & 3.06 & 1.18 & 7.44 & 1.93 & 7.14 & 1.58 & 7.28 & 2.07 & 8.37 \\
 & \centering SB-o  & \textbf{0.15} & \textbf{0.22} & \textbf{0.10} & \textbf{0.16} & \textbf{0.10} & \textbf{0.14} & \textbf{0.13} & \textbf{0.17} & \textbf{0.12} & \textbf{0.19} & \textbf{0.16} & \textbf{0.28} & \textbf{0.21} & \textbf{0.29} & \textbf{0.21} & \textbf{0.26} & \textbf{0.25} & \textbf{0.41} \\
\hline
\centering\multirow{2}{*}{$\num{8000}$} & \centering NC & 1.62 & 7.16 & 1.98 & 4.97 & 0.62 & 2.52 & 1.05 & 6.24 & 1.79 & 4.27 & 1.40 & 8.22 & 3.29 & 9.23 & 3.04 & 10.36 & 2.45 & 8.94 \\
 & \centering SB-o  & \textbf{0.11} & \textbf{0.33} & \textbf{0.12} & \textbf{0.21} & \textbf{0.10} & \textbf{0.14} & \textbf{0.12} & \textbf{0.15} & \textbf{0.13} & \textbf{0.21} & \textbf{0.17} & \textbf{0.23} & \textbf{0.20} & \textbf{0.33} & \textbf{0.26} & \textbf{0.36} & \textbf{0.28} & \textbf{0.35} \\
\hline
\centering\multirow{2}{*}{$\num{10000}$} & \centering NC & 1.70 & 10.86 & 1.56 & 7.75 & 1.09 & 5.00 & 1.56 & 9.06 & 2.46 & 8.77 & 3.18 & 9.55 & 5.83 & 12.19 & 5.92 & 15.18 & 8.96 & 23.97 \\
 & \centering SB-o  & \textbf{0.12} & \textbf{0.20} & \textbf{0.10} & \textbf{0.26} & \textbf{0.10} & \textbf{0.23} & \textbf{0.12} & \textbf{0.17} & \textbf{0.14} & \textbf{0.19} & \textbf{0.19} & \textbf{0.29} & \textbf{0.21} & \textbf{0.29} & \textbf{0.26} & \textbf{0.33} & \textbf{0.33} & \textbf{0.44} \\
    \hline
\end{tabularx}
\caption{Median and 90\% quantile $Q_{.9}$ of running times in seconds of NC, SB-r and SB-o on random graphs of all tested sizes for $k \in \{200, 500, 1000\}$.}
\end{sidewaystable}


\begin{sidewaysfigure}[tb]
\centering
\includegraphics[height=0.65\textwidth]{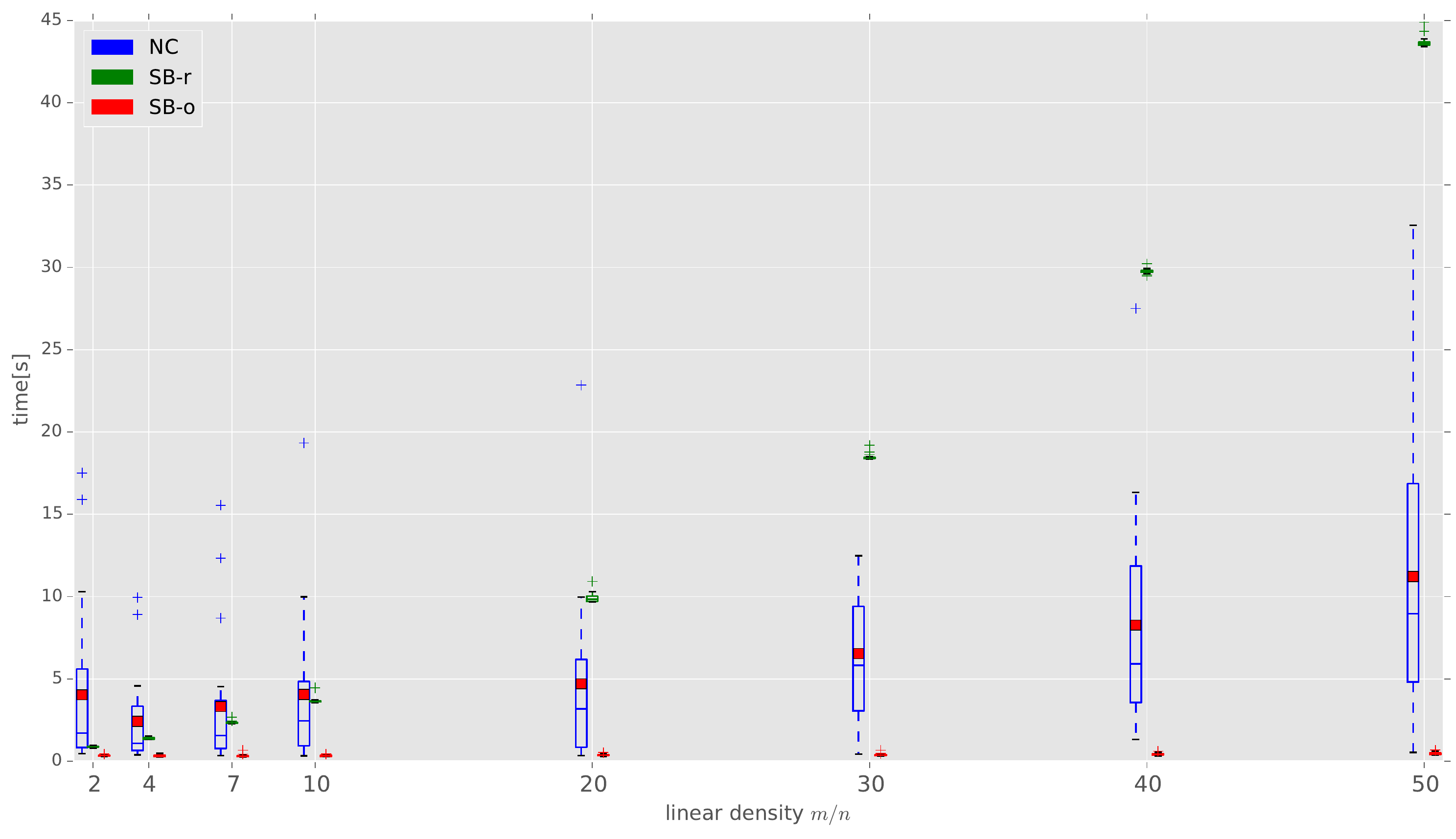}
\caption{Boxplots of running times in seconds for random graphs with $\num{10000}$ nodes and $k = \num{2000}$. Plus signs represent outliers. A red square marks the mean, but was omitted for the SB algorithms as they would completely cover their corresponding boxes.}
\label{fig:sprand-boxplot}
\end{sidewaysfigure}

\begin{sidewaysfigure}[tb]
\centering
\includegraphics[height=0.65\textwidth]{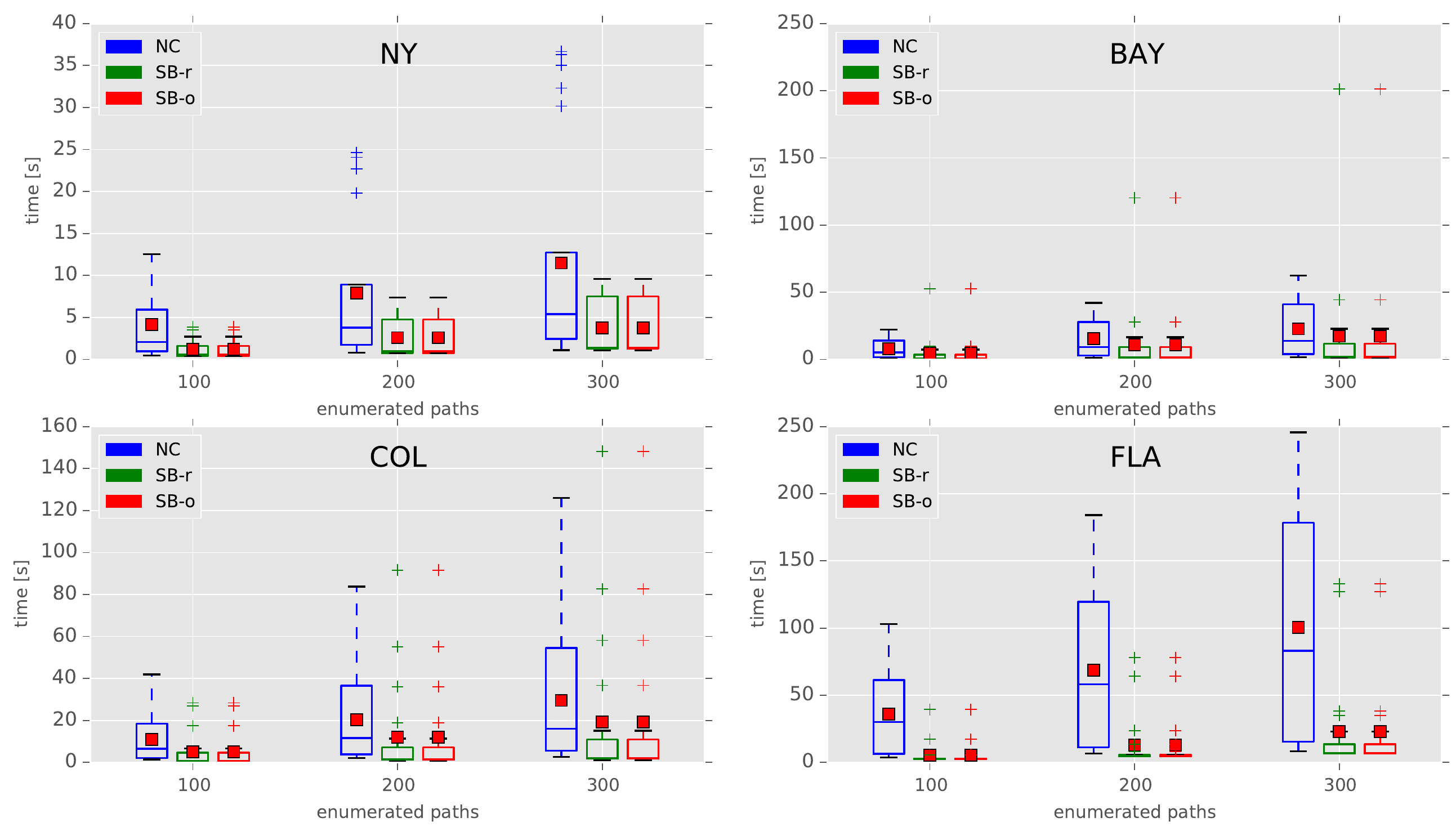}
\caption{Boxplots of running times in seconds for TIGER graphs. Plus signs represent outliers. A red square marks the mean. The interquartile range of SB on FLA for $k \in \{100, 200\}$ is so small that the boxes appear as lines.}
\label{fig:large-cities}
\end{sidewaysfigure}

\end{document}